\documentclass{article}
\usepackage{ijcai23}

\usepackage{times}
\usepackage{soul}
\usepackage{url}
\usepackage[hidelinks]{hyperref}
\usepackage[utf8]{inputenc}
\usepackage[small]{caption}
\usepackage{graphicx}
\usepackage{amsmath}
\usepackage{amsthm}
\usepackage{booktabs}
\usepackage{algorithm}
\usepackage{algorithmic}
\usepackage[switch]{lineno}

\usepackage{amssymb}
\usepackage{framed}
\usepackage{tikz}
\usepackage{subfigure}
\usepackage{multirow}
\usepackage{pifont}
\usepackage{ulem}
\newtheorem{theorem}{Theorem}
\newtheorem{definition}{Definition}

\title{Truthful and Stable One-sided Matching on Networks}

\author{
Tianyi Yang
\and
Yuxiang Zhai\and
Dengji Zhao\and
Xinwei Song\And
Miao Li
\affiliations
ShanghaiTech University
\emails
\{yangty2, zhaiyx, zhaodj, songxw, limiao\}@shanghaitech.edu.cn
}

\begin{document}

\maketitle

\begin{abstract}
Mechanism design on social networks is a hot research direction recently, and we have seen many interesting results in auctions and matching. Compared to the traditional settings, the new goal of the network settings is that we need to design incentives to incentivize the participants of the game to invite their neighbors on the network to join the game. This is challenging because they are competing for something (e.g., resources or matches) in the game. In one-sided matching, especially house exchange, the well-known unique truthful, stable and optimal solution called Top Trading Cycle (TTC) cannot achieve the new goal. Existing works have tried to add constraints on TTC to obtain the incentive, but it only works in trees and it does not guarantee any stability. In this paper, we move this forward and propose the first mechanism called Leave and Share (LS) which not only achieves the goal in all networks but also gives the most stable solution in the new settings. In terms of optimality, as it is impossible to achieve it in any network, we conduct simulations to compare it with the extensions of TTC.

\end{abstract}

\section{Introduction}

Incentivizing agents to invite new agents via their social connections is a new trend in mechanism design. Different from traditional static settings, the agents' connections are specifically considered and utilized by the mechanisms to enlarge the market~\cite{DBLP:conf/atal/Zhao21}. This is achieved by incentivizing the agents who are already in the game to invite their neighbors to join the game. In most games, a larger market contributes to a more desirable outcome. Particularly, in one-sided matching, more participants may lead to a more satisfied matching. 


However, inviting more agents is not always better for the inviters. For example, A starts a matching game, and she invites her neighbors B and C to join in, and A prefers both B's and C's items. If in the end, B exchanges with C, and A does not exchange with any of them. However, if B is not in the game, A will exchange with C. In this example, A is not incentivized to invite B to this game.

To remove A's hesitation to invite all her neighbors, 
we should guarantee that the inviters' match is not getting worse after inviting their neighbors. 
In traditional one-sided matching, the well-known Top Trading Cycle (TTC) mechanism gives the unique truthful, stable and optimal solution~\cite{shapley1974cores,ma1994strategy}. However, TTC failed to incentivize the participants to invite others, because an invitee might compete with her inviters for the same match~\cite{DBLP:conf/atal/KawasakiWTY21}. 
In order to achieve the incentive, we may add constraints on TTC such as each agent can only exchange with certain agents or the network has to be trees~\cite{zheng2020barter,DBLP:conf/atal/KawasakiWTY21}. 

The restrictions on TTC indeed obtained the invitation incentive, but they limited the participants' matching choices and contradict the purpose of enlarging the game. The restrictions simply do not allow an invitee to have any chance to compete with her inviters, 
which also forbid a better allocation for those invitees who will not harm their inviters. Therefore, we still have space to further improve the matching by relaxing the restrictions. The challenge is that we cannot know whether an invitee will bring harm to her inviters or not without fixing a matching mechanism. That said, it seems that a natural way to find a better mechanism is through trial and error. However, the space and description of the matching mechanism are exponential in the number of participants, which is not doable in practice.



In this paper, we move this effort forward by proposing a new way of designing a better matching by relaxing a rather restricted mechanism. Our mechanism is based on a very restricted mechanism, called Swap With Neighbors(SWN), which only allows each agent to choose her preferred agent/item among her neighbors (also a restricted version of TTC). To relax the restriction, we allow matched agents to share their unmatched neighbors with others. By doing so, the later matched agents will have more choices, which is essential to improve their matching. We call this new mechanism Leave and Share (LS). This seems straightforward, but it is not because all agents would prefer matching later in order to get more shared choices from matched agents. Therefore, we need to carefully design the matching process to both utilize the shared benefits and also prevent new manipulations.


Finally, the LS mechanism can work for all networks (remove the structure restriction of the previous work). More importantly, we also show that it gives the most stable solution under networks (which is also firstly proved here). In summary, our contributions advance the state of the art in the following ways:

\begin{itemize}

    \item We prove the impossibilities of incentivizing invitation under the condition of optimality (stability). For stability, we extend the definition and prove what we can (cannot) achieve with invitation incentives. 
    \item We propose the LS mechanism to achieve both invitation incentive and the extended stability for the first time. LS equals TTC if the network is a complete graph. 
    \item To see the matching improvement under LS, 
    we conduct simulations to compare LS with TTC extensions. The results indicate that the sharing feature of LS indeed improves the matching in expectation in all settings under consideration.
\end{itemize}

\subsection{Related Work}

Mechanism design over social networks is a hot research topic, where agents’ social connections and interactions are considered by the mechanism. 
One popular way to utilize the connections is to attract new participants, which has made a significant progress in auctions and cooperative games \cite{DBLP:conf/aaai/LiHZZ17,li2022diffusion,zhang2022incentives}. In auctions, the main technique proposed to design the invitation incentives is to allow buyers to gain the social welfare increase due to their invitation. In cooperative games, invitees need to share their contributions with their inviters to get invited. Their methods rely on transferable utilities, which is not possible in matching. 

Specifically, for one-sided matching over social networks, traditional solutions like TTC cannot be directly applied to incentivize participants' invitation. To make TTC work in the network setting, \citeauthor{DBLP:conf/atal/KawasakiWTY21} \shortcite{DBLP:conf/atal/KawasakiWTY21} presented a modified TTC under tree networks only. Their modification restricted each participant's choices to her parent and her subtree. Their extension incentivized invitation by disabling more satisfiable matchings. \citeauthor{DBLP:conf/ijcai/GourvesLW17} \shortcite{DBLP:conf/ijcai/GourvesLW17} also studied one-sided matching in social networks, but they focused on static social connections and ignored agents' strategic behaviors.
To evaluate the efficiency of a matching mechanism, some cardinal methods are well-studied in the traditional setting \cite{PO_in_house_allocation_problems,cardinal_pref}. Motivated by these work, we also apply a cardinal method to evaluate 
our mechanism. 

Besides standard one-sided matching, \citeauthor{tenants_social_network} \shortcite{tenants_social_network} also studied a variant of house allocation problem \cite{abdulkadirouglu1999house} under the tree-structured networks. Moreover, \citeauthor{ijcai2022-27} \shortcite{ijcai2022-27} investigated two-sided matching over social networks, where they presented a series of impossibilities to show the hardness of the problem, and designed mechanisms under the assumption that one side is known. The goal of these papers is to extend the traditional solutions to get the invitation incentives required in the corresponding network settings.

\section{The Model}\label{sec:md}

We consider a one-sided matching problem in a social network denoted by an undirected graph $G=(N,E)$, which contains $n$ agents $N=\{1,\dots,n\}$. Each agent $i\in N$ is endowed with an indivisible item $h_i$ and $H=\{h_1,\dots,h_n\}$ is the set of all agents' items. We define agent $i$ as $j$'s neighbor if there is an edge $e\in E$ between agent $i$ and $j$, and let $r_i\subseteq N$ be $i$'s neighbor set. 

Each agent $i\in N$ has a strict preference $\succ_i$ over $H$. $h\succ_i h'$ means $i$ prefers $h$ to $h'$ and we use $\succeq_i$ to represent the weak preference. Denote agent $i$'s private type as $\theta_i=(\succ_i,r_i)$ and $\theta = (\theta_1,\cdots,\theta_n)$ as the type profile of all agents. Let $\theta_{-i}$ be the type profile of all agents except for agent $i$, then $\theta$ can be written as $(\theta_i,\theta_{-i})$. Let $\Theta$ be the type profile space of all agents. Similarly, we have $\Theta = (\Theta_i,\Theta_{-i})$. 

In a matching mechanism, each agent is required to report her type (reporting neighbor set is treated as inviting neighbors in practice). We denote agent $i$'s reported type as $\theta'_i = (\succ'_i,r'_i)$, where $\succ'_i$ is the reported preference and $r'_i\subseteq r_i$ is the reported neighbor set. Let $\theta'=(\theta'_1,\cdots,\theta'_n)$ be the reported type profile of all agents.


\begin{definition}

A one-sided matching mechanism is defined by an allocation policy $\pi = (\pi_i)_{i\in N}$, where $\pi_i:\Theta \to H$ satisfies for all $\theta \in \Theta$, for all $i$, $\pi_i(\theta) \in H$, and $\pi_i(\theta) \neq \pi_j(\theta)$ for all $i \neq j$.
\end{definition}

Different from traditional settings, we assume only a subset of the agents are initially in the game (e.g., one agent initiated a matching with her neighbors). Without loss of generality, suppose an agent set $N_0 \subseteq N$ contains the initial participants in the matching. The others need the existing participants' invitation to join the game. As the invitation process is modeled by reporting their neighbors, we define the qualified participants by their reported types.

For a given report profile $\theta'$, we generate a directed graph $G(\theta')=(N(\theta'),E(\theta'))$, where edge $\langle i,j\rangle \in E(\theta')$ if and only if $j \in r_i'$. Under $\theta'$, we say agent $i$ is qualified if and only if there is a path from any agent in $N_0$ to $i$ in $G(\theta')$. That is, $i$ can be properly invited by the invitation chain from agent set $N_0$. Let $Q(\theta')$ be the set of all qualified agents under $\theta'$. Then 
the matching mechanism can only use $Q(\theta')$.


\begin{definition}
A diffusion one-sided matching mechanism in social networks is a one-sided matching mechanism, $\pi = (\pi_i)_{i\in N}$, such that for all reported type profile $\theta'$, it satisfies:
\begin{enumerate}
    \item for all unqualified agents $i\notin Q(\theta')$, $\pi_i(\theta') = h_i$.
    \item for all qualified agents $i\in Q(\theta')$, $\pi_i(\theta')$ is independent of the reports of all unqualified agents.
 
\end{enumerate}
    
\end{definition}

The difference between a diffusion one-sided matching and the matching defined in Definition 1 is that the participants can affect the qualification of other participants. If a participant changes her reported neighbor set, the qualified agent set may change. This is the challenge of this setting.

Next, we define two desirable properties for diffusion one-sided matching mechanisms: individual rationality and incentive compatibility. Intuitively, individual rationality requires that for each agent, reporting her type truthfully guarantees that she gets an item no worse than her own.

\begin{definition}[Individual Rationality (IR)]
A diffusion one-sided matching mechanism $\pi$ is individually rational if for all $i\in N $, all $\theta_i \in \Theta_i$, and all $\theta'_{-i}\in \Theta_{-i}$, we have $\pi_i(\theta_i,\theta_{-i}')\succeq_i h_i$.

\end{definition}

For incentive compatibility, it means reporting type truthfully is a dominant strategy for each agent.

\begin{definition}[Incentive Compatibility (IC)]
A diffusion one-sided matching mechanism $\pi$ is incentive compatible if for all $i\in N$, all $\theta'_{-i}\in \Theta_{-i}$ and all $\theta_i, \theta'_i\in \Theta_i$, we have $\pi_i(\theta_i,\theta'_{-i})\succeq_i \pi_i(\theta'_i,\theta'_{-i})$.

\end{definition}

To evaluate the performance of a matching mechanism, an important metric is called Pareto optimality.


\begin{definition}[Pareto Optimality (PO)]
A mechanism $\pi$ is Pareto optimal if for all type profile $\theta$, there is no other allocation $\pi'(\theta)$ such that for each agent $i$, $\pi'_i(\theta) \succeq_i \pi_i(\theta)$, and there exists at least one agent $j$, $\pi'_j(\theta) \succ_j \pi_j(\theta)$.
\end{definition}

Another metric is stability. A matching is stable if there does not exist any subset of agents who can deviate from the matching and match among the subset to make no one worse off, but at least one better off (this is called a blocking coalition). In the setting without networks, any subset of players can form a blocking coalition. However, in the network setting, they should know each other before they can form a coalition. Therefore, we assume that the blocking coalition in our setting is at least connected.


\begin{definition}[Blocking Coalition]
Given an allocation $\pi(\theta)$, we say a set of agents $S \subseteq N$ (with item set $H_S \subseteq H$) is a blocking coalition for $\pi(\theta)$ if $S$ forms a \textbf{connected component} in $G(\theta)$ and there exists an allocation $z(\theta)$ such that for all $i \in S, z_i(\theta) \in H_S$ and $z_i(\theta) \succeq_i \pi_i(\theta)$ with at least one $j\in S$ such that $z_j(\theta) \succ_j \pi_j(\theta)$. 
\end{definition}

\begin{definition}[Stability]
We say a mechanism $\pi$ is stable if for all type profiles $\theta$, there is no blocking coalition for $\pi(\theta)$.
\end{definition}

\section{Impossibility Results}\label{sec:ir}

In this section, we discuss the impossibility results in the network setting presented in Table 1.

\begin{table}[t]
\centering
\begin{tabular}{|c|c|c|c|c|}
    \hline
     & PO+IR & Stable & Stable-WCC & Stable-CC \\ \hline
    IC &\ding{53} &\ding{53} & \ding{53} & \ding{51} \\ \hline
\end{tabular}

\label{tab:impossibility}
\caption{The coexistence of IC with other properties.}
\end{table}

\begin{theorem}[Impossibility for PO, IC and IR]
Given a social network with no less than three agents, no diffusion matching mechanism is PO, IC and IR.
\end{theorem}

\begin{proof}
In the example shown in Figure \ref{fig:sta1}, the only PO and IR allocations are $a_3$ and $a_6$. For the former, agent 1 can misreport her preference as $h_3 \succ_1 h_1 \succ_1 h_2$. Under agent 1’s misreport, the only PO and IR allocation will be $a_6$, and 1 reaches a better allocation under $a_6$ thus violating IC. For the latter, agent 2 can misreport her neighbor set as \{1\} and disqualify agent 3. In this way, the only PO and IR allocation is $a_3$, 2 reaches a better allocation which also violates IC. Hence, no mechanism under the social network setting with no less than 3 agents is PO, IC and IR. 
\end{proof}

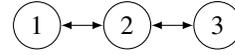
\begin{figure}[t]
\centering

\begin{tikzpicture}[scale=0.3, sibling distance=5em,
  every node/.style = {scale=1, shape=circle, draw, align=center},
    outline/.style={draw=#1,thick,fill=#1!100}]
  every draw/.style = {scale=1}
  \node[] (node1) at (0,0) {1};
  \node[] (node2) at (4,0) {2};
  \node[] (node3) at (8,0) {3};
  \draw[latex-latex] (node1)--(node2);
  \draw[latex-latex] (node2)--(node3);
  
\end{tikzpicture}

\caption{An example of social networks. }
\label{fig:sta1}
\end{figure}


\begin{table}[t]
\small
\centering
\setlength{\tabcolsep}{1pt}{
 \begin{tabular}{|c|c|c|c|c|c|c|}
 \hline
Preference & Allocation & PO+IR & TTC & Stable & SCC & SWCC \\ \hline
\multirow{2}{*}{$h_3 \succ_1 h_2 \succ_1 h_1$} & $a_1$:($h_1,h_2,h_3$) &  \ding{53}  &   \ding{53}  & \ding{53}  & \ding{53} & \ding{53} \\ \cline{2-7} 
& $a_2$:($h_1,h_3,h_2$) &  \ding{53}  &   \ding{53}  & \ding{53}  & \ding{53} & \ding{53} \\ \cline{2-7} 
\multirow{2}{*}{$h_1 \succ_2 h_2 \succ_2 h_3$} & $a_3$:($h_2,h_1,h_3$) &  \ding{51} & \ding{53}  & \ding{51}  & \ding{51} & \ding{51}  \\ \cline{2-7} 
& $a_4$:($h_2,h_3,h_1$) &  \ding{53}  &   \ding{53}  & \ding{53}  & \ding{53} & \ding{53} \\ \cline{2-7} 
\multirow{2}{*}{$h_1 \succ_3 h_3 \succ_3 h_2$} & $a_5$:($h_3,h_1,h_2$) &  \ding{53}  &  \ding{53}  & \ding{53}  & \ding{53} & \ding{53} \\ \cline{2-7} 
& $a_6$:($h_3,h_2,h_1$) &  \ding{51} & \ding{51}  & \ding{51}  & \ding{51} & \ding{51} \\ \hline \hline 
\multirow{2}{*}{\textcolor{red}{$h_3 \succ_1' h_1 \succ_1' h_2$}} & $a_1$:($h_1,h_2,h_3$) &  \ding{53}  &   \ding{53}  & \ding{53}  & \textcolor{red}{\ding{51}} & \ding{53} \\ \cline{2-7} 
& $a_2$:($h_1,h_3,h_2$) &  \ding{53}  &   \ding{53}  & \ding{53}  & \ding{53} & \ding{53} \\ \cline{2-7} 
\multirow{2}{*}{$h_1 \succ_2 h_2 \succ_2 h_3$} & $a_3$:($h_2,h_1,h_3$) &  \ding{53} & \ding{53}  & \textcolor{red}{\ding{53}}  & \textcolor{red}{\ding{53}} & \textcolor{red}{\ding{53}}  \\ \cline{2-7} 
& $a_4$:($h_2,h_3,h_1$) &  \ding{53}  &   \ding{53}  & \ding{53}  & \ding{53} & \ding{53} \\ \cline{2-7} 
\multirow{2}{*}{$h_1 \succ_3 h_3 \succ_3 h_2$} & $a_5$:($h_3,h_1,h_2$) &  \ding{53}  &  \ding{53}  & \ding{53}  & \ding{53} & \ding{53} \\ \cline{2-7} 
& $a_6$:($h_3,h_2,h_1$) &  \ding{51} & \ding{51}  & \ding{51}  & \ding{51} & \ding{51} \\ \hline
\end{tabular}
}

\caption{All possible allocations of Figure \ref{fig:sta1} and whether an allocation satisfies specific properties. If agent 1 misreports, the allocation satisfies Stable-CC (SCC) will be $a_1$ and $a_6$, instead of $a_6$ only as under stable or Stable-WCC (SWCC).}
\label{tab:theta1}
\vspace{-1.0em}
\end{table} 

A similar result also holds even if the neighbor relationship is asymmetric\cite{DBLP:conf/atal/KawasakiWTY21}.

\begin{theorem}[Impossibility for stability and IC]
Given a social network with no less than three agents, no diffusion matching mechanism is stability and IC.
\label{the2}
\end{theorem}

\begin{proof}
Consider the example given in Figure \ref{fig:sta1}, the stable allocations are identical to the PO and IR ones (under the two type profile in Table \ref{tab:theta1}). With the same strategic misreport, it can be concluded that stability is incompatible with IC.
\end{proof}


To seek for an achievable stability in social networks, we should further restrict the blocking coalitions. In the traditional setting, since there are no constraints on social connections, the agents can be viewed as fully connected. Then, any blocking coalition is a complete component. Therefore, we require the blocking coalitions to be complete components.


 \begin{definition}[Blocking Coalition under Complete Components]
 Given an allocation $\pi(\theta)$, we say a set of agents $S \subseteq N$ (with item set $H_S \subseteq H$) is a \textbf{blocking coalition under complete components} for $\pi(\theta)$ if $S$ forms a complete component in $G(\theta)$ and there exists an allocation $z(\theta)$ such that for all $i \in S, z_i(\theta) \in H_S$ and $z_i(\theta) \succeq_i \pi_i(\theta)$ with at least one $j\in S,  z_j(\theta) \succ_j \pi_j(\theta)$. 
 \end{definition}

 \begin{definition}[Stability under Complete Components (Stable-CC)]
We say a mechanism $\pi$ is \textbf{stable under complete components} if for all type profiles $\theta$, there is no blocking coalition under complete components for $\pi(\theta)$.
 \end{definition}

In this paper, we will design a matching mechanism that satisfies IC, IR and Stable-CC. This extended stability looks rather restricted, is it possible to make a slight relaxation on the complete component? In fact, it is not achievable even if we just remove a single edge from a complete component.

\begin{definition}[Nearly Complete Component]
We call a connected graph $G=(V,E)$ a nearly complete component if $|E| = \frac{|V|\cdot(|V|-1)}{2} - 1$. 
\end{definition}

The gap between a complete component and a nearly complete component is only one edge. We define a new stability under nearly complete components and prove that it is impossible to coexist with IC. 

\begin{definition}[Blocking Coalition under Weakly Complete Components]
Given an allocation $\pi(\theta)$, we say a set of agents $S \subseteq N$ (with item set $H_S \subseteq H$) is a \textbf{blocking coalition under weakly complete components} for $\pi(\theta)$ if $S$ forms a nearly complete component or a complete component in $G(\theta)$ and there exists an allocation $z(\theta)$ such that for all $i \in S, z_i(\theta) \in H_S$ and $z_i(\theta) \succeq_i \pi_i(\theta)$ with at least one $j\in S,  z_j(\theta) \succ_j \pi_j(\theta)$. 
\end{definition}

 \begin{definition}[Stability under Weakly Complete Components (Stable-WCC)]
We say a mechanism $\pi$ is \textbf{stable under weakly complete components} if for all type profiles $\theta$,  there is no blocking coalition under weakly complete components for $\pi(\theta)$.
 \end{definition}


\begin{theorem}[Impossibility for Stable-WCC and IC]
Given a social network with no less than three agents, no diffusion matching mechanism is Stable-WCC and IC.
\end{theorem}

\begin{proof}
Recall the example given in Figure \ref{fig:sta1}, the possible blocking coalitions and allocations under Stable-WCC is identical to the example used in the proof of Theorem~\ref{the2}. Therefore, the same reasoning applies here. 
\end{proof}

\section{The Mechanism}\label{sec:mc}

Before we introduce our mechanism, we first define the Top Trading Cycle and its extensions.
\begin{definition}[Top Trading Cycle]
For a given $G(\theta')$, \uline{construct a directed graph by letting each agent point to the agent who has her favorite item remaining in the matching.} There is at least one cycle. For each cycle, allocate the item to the agent who points to it and remove the cycle. Repeat the process until there is no agent left. 

\end{definition}



TTC cannot ensure IC in the new setting, one trivial extension is called Swap With Neighbors (SWN), which only allows agents to swap with their neighbors. Intuitively, agents can only get allocated items from those who they invite and if they do not invite their neighbors, they will have no opportunity to get allocated items they prefer. Formal proof will be given in the appendix.


\begin{definition}[Swap With Neighbors]
For a given $G(\theta')$, \uline{construct a directed graph by letting each agent point to her favorite item among herself and her neighbors remaining in the matching.} There is at least one cycle. For each cycle, allocate the item to the agent who points to it and remove the cycle. Repeat the process until there is no agent left. 
\end{definition}
Another attempt is to restrict the network to trees and allow each agent to swap with her neighbors and subtree
\cite{DBLP:conf/atal/KawasakiWTY21}. Let's call it Swap With Children (SWC). 

\begin{definition}[Swap With Children]
For a given $G(\theta')$, \uline{construct a directed graph by letting each agent points to her favorite item among herself, her neighbors, and her descendants remaining in the matching.} There is at least one cycle. For each cycle, allocate the item to the agent who points to it and remove the cycle. Repeat the process until there is no agent left.  
\end{definition}


Both SWN and SWC avoid competition by restricting matching choices, which is not our goal to enlarge the market. We propose a new mechanism called Leave and Share (LS) which satisfies IC, IR and Stable-CC in all networks. Table 3 shows the difference of these mechanisms.

\begin{table}[t]
\centering
\begin{tabular}{|c|c|c|c|}
    \hline
    \multirow{2}{*}{Mechanism} & \multirow{2}{*}{Stable-CC} & \multicolumn{2}{c|}{IC} \\ \cline{3-4} 
     & & Trees & All Networks\\ \hline
    TTC &\ding{53} & \ding{53} & \ding{53} \\ \hline
    \textcolor{red}{SWN} & \textcolor{red}{\ding{51}} & \textcolor{red}{\ding{51}} & \textcolor{red}{\ding{51}} \\ \hline
    SWC &\ding{51}(Trees) & \ding{51} & \ding{53} \\ \hline
    \textcolor{red}{LS}  &\textcolor{red}{\ding{51}} & \textcolor{red}{\ding{51}} & \textcolor{red}{\ding{51}} \\ \hline
\end{tabular}
\label{tab:mechanism}
\caption{Comparison on one-sided matching mechanisms over social networks.}
\vspace{-1.0em}
\end{table}

\subsection{Leave and Share}

Leave and Share uses SWN as a base and adds a natural sharing process to enlarge agents' selection space, trying to provide a better allocation. Firstly, agents are matched by rounds in a protocol that resembles SWN under a strategy-proof order. This guarantees that inviters are not worse off. Then, we share the neighbors of the left agents in this round by connecting their neighbors to each other, thus their neighbors can have new neighbors in the next round. This dynamic neighbor set update comes naturally because a matched cycle does not care how the remaining neighbors will be matched. Also, their remaining neighbors cannot prevent this sharing, and neither can the other remaining agents.

\begin{figure}[t]
\centering

\begin{tikzpicture}[scale=0.3, sibling distance=5em,
  every node/.style = {scale=1, shape=circle, draw, align=center},
    outline/.style={draw=#1,thick,fill=#1!100}]
  every draw/.style = {scale=1}
  \node[] (node1) at (0,-3.5) {1};
  \node[] (node2) at (0,0) {2};
  \node[] (node3) at (4,0) {3};
  \node[] (node4) at (8,0) {4};
  \node[] (node5) at (12,0) {5};
  \node[] (node6) at (12,-3.5) {6};
  \draw[latex-latex] (node1)--(node2);
  \draw[latex-latex] (node2)--(node3);
  \draw[latex-latex] (node3)--(node4);
  \draw[latex-latex] (node4)--(node5);
  \draw[latex-latex] (node5)--(node6);
  \draw[-latex,red] (node3).. controls(6,1) ..(node4);
  \draw[-latex,red] (node4).. controls(6,-1) ..(node3);
  \draw[-latex,red] (node2).. controls(6,2)..(node5);
  \draw[-latex,red] (node5).. controls(6,-2)..(node2);
  \draw[-latex,red] (node1).. controls(6,-2) ..(node6);
  \draw[-latex,red] (node6).. controls(6,-3.5) ..(node1);

\end{tikzpicture}

\caption{Preferences are $h_6 \succ_1 h_1 \succ_1 \cdots, \ \ \  h_5 \succ_2 h_2 \succ_2 \cdots, \ \ \ h_4 \succ_3 h_3 \succ_3 \cdots, \ \ \ h_3 \succ_4 h_4 \succ_4 \cdots, \ \ \ h_2 \succ_5 h_5 \succ_5 \cdots, \ \ \ h_1 \succ_6 h_6 \succ_6 \cdots$. One allocation for the agents $1$ to $6$ is $( h_6,h_5,h_4,h_3,h_2,h_1 )$.}
\label{fig:value_ls}
\end{figure}
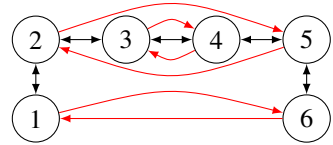

To see the value of our mechanism, consider the example given in Figure~\ref{fig:value_ls}, where only agents 3 and 4 can exchange with each other in both SWN and SWC. The rest of the agents will end up with their own items. However, agents 3 and 4 will not block the exchange for agents 2 and 5 once they get their preferred items. After agents 3 and 4 are matched and \textbf{Leave}, we \textbf{Share} their remaining neighbors then agents 2 and 5 can swap. Similarly, after agents 2 and 5 leave, agents 1 and 6 can be matched as well. The process of Leave and Share is the name and core of our mechanism.


Before formalizing our mechanism, we introduce two notations to simplify the description. 

\begin{definition}
Given a set $A \subseteq N$, we say $f_i(A)= j\in A$ is $i$'s favorite agent in $A$ if for any agent $k\in A, h_j\succeq'_i h_{k}$. 
\end{definition}

\begin{definition}
An ordering of agents is a one-to-one function $\mathcal {P}:\mathbb {N}^{+}\to N$, where agent $\mathcal {P}(i)$ is the $i^{th}$ agent in the ordering. Agents in $\mathcal{P}$ are sorted in ascending order by the length of the shortest path from agent set $N_0$ to them. Especially, for any agent $i \in N_0$, its shortest path length
is $0$. When multiple agents have the same length of the shortest path, we use a random tie-breaking.

\end{definition}

\begin{framed}
 \noindent\textbf{Leave and Share (LS)}
  
 \begin{enumerate}
  \item Initialize $N_{out}=\emptyset$ and an empty stack $S$. Define the top and bottom of $S$ as $S_{top}$ and $S_{bottom}$ respectively, and let $R_{i}=r_{i}'\cup \{ S_{bottom} , i\}$.
  \item While $N_{out} \neq N$:
  \begin{enumerate}
    \item Find the minimum $t$ such that $\mathcal{P}(t) \notin N_{out}$. Push $\mathcal{P}(t)$ into $S$. 
    
      
    \item While $S$ is not empty:
        \begin{enumerate}
            \item While $f_{S_{top}}(R_{S_{top}}) \notin S$, push $f_{S_{top}}(R_{S_{top}})$ into $S$.
            \item Pop off all agents from $S_{top}$ to $f_{S_{top}}(R_{S_{top}})$, who already formed a trading cycle $C$ following their favorite agents. Allocate each agent $i\in C$ the item $h_{f_i(R_i)}$ . Add $C$ to $N_{out}^t$.
            \item Update the neighbor set of $C$'s remaining neighbors by removing $C$, i.e., for all $j\in \bigcup_{i\in C} r_i' \setminus N_{out}^t$, set $r_j' = r_j' \setminus C$. 
         \end{enumerate}
    \item Add $N_{out}^t$ to $N_{out}$. Let all remaining neighbors of $N_{out}^t$ connect with each other, i.e., they become neighbors of each other. That is, let $X = \bigcup_{i\in N_{out}^t} r_i' \setminus N_{out}^t$ and for all $j\in X$, set $r_j' = r_j' \cup X$. 
    
    \end{enumerate}
  
 \end{enumerate}
\end{framed}
In LS, we first define an order $\mathcal {P}$ which depends on each agent's shortest distance to the initial agent set. Under this order, the first while loop (step $2$) guarantees that the agent pushed into the stack is the remaining agent with the smallest order, and all agents are matched (including self-match) in the end. A new round begins each time the stack empties. 

In the Leave stage, each agent that is pushed into the stack pushes her (current) favorite neighbor into the stack (step (a)). If her favorite agent is already in the stack, we pop all the agents between herself and her favorite agent to form a trading cycle. Specially, we allow the agent to choose the agent at the bottom of the stack as her favorite, which leads to the pop of all the agents in the stack (step (b)). 

Once the stack is empty, the mechanism enters the Share stage and updates the neighbor set of the remaining agents (step (c)). All the neighbors of the left agents become new neighbors to each other. In the next Leave stage, they can choose their favorite neighbors in a larger neighbor set. 



We illustrate how LS executes by an example. Consider the social network in Figure~\ref{fig:sn}(a). The ordering is given as $\mathcal{P}=(1,2,3,4,5,6,7,9,8)$, $N_0={1}$. The type profile $\theta$ is given in Table~\ref{tab:theta}. The changes of $N_{out}^t$ and $r_i'$ in each turn are shown in Table~\ref{tab:neigh}.  Figure~\ref{fig:sn}(b) to Figure~\ref{fig:sn}(d) show the process of Leave and Share,  
which runs as follows.

\begin{table}[t]
    \centering

    \begin{tabular}{|c|c|l|}
\hline
i & $\succ_i$  & $r_i$ \\ \hline
1 & $h_2 \succ h_4 \succ h_3 \succ h_1 \succ \cdots$ & 2 \\ \hline
2 & $h_4 \succ h_3 \succ h_5 \succ h_2 \succ \cdots$ & 1,3,4 \\ \hline
3 & $h_8 \succ h_4 \succ h_1 \succ h_3 \succ \cdots$ & 2,7 \\ \hline
4 & $h_5 \succ h_1 \succ h_3 \succ h_4 \succ \cdots $ & 2,5,6,9 \\ \hline
5 & $h_2 \succ h_6 \succ h_9 \succ h_5 \succ \cdots$ & 4,6,8 \\ \hline
6 & $h_4 \succ h_1 \succ h_8 \succ h_6 \succ \cdots$ & 4,5 \\ \hline
7 & $h_8 \succ h_3 \succ h_6 \succ h_7 \succ \cdots$ & 3 \\ \hline
8 & $h_9 \succ h_5 \succ h_1 \succ h_8 \succ \cdots$ & 5 \\ \hline
9 & $h_7 \succ h_4 \succ h_6 \succ h_9 \succ \cdots$ & 4 \\ \hline
\end{tabular}

\caption{The preference and neighbors of each agent. We omit those items which rank lower than their own items.}
\label{tab:theta}
\end{table}
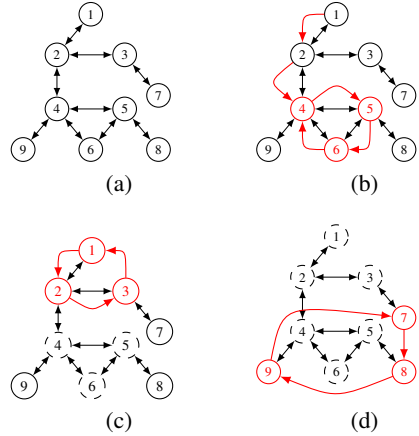
\begin{figure}[ht]

\centering

\subfigure[]{
\centering
\begin{minipage}[t]{0.35\linewidth}

\begin{tikzpicture}[scale=0.18, sibling distance=0em,
  every node/.style = {scale=0.5, shape=circle, draw, align=center},
    outline/.style={draw=#1,thick,fill=#1!100}]
  every draw/.style = {scale=2}
  
  \node (node1) at (9,10) {1};
  \node (node2) at (6.5,7) {2};
  \node (node3) at (11.5,7) {3};
  \node (node4) at (6.5,3) {4};
  \node (node5) at (11.5,3) {5};
  \node (node6) at (9,0) {6};
  \node (node7) at (14,4) {7};
  \node (node8) at (14,0) {8};
  \node (node9) at (4,0) {9};
  \draw[latex-latex,very thin] (node1)--(node2);
  \draw[latex-latex,very thin] (node2)--(node3);
  \draw[latex-latex,very thin] (node2)--(node4);
  \draw[latex-latex,very thin] (node4)--(node5);
  \draw[latex-latex,very thin] (node4)--(node6);
  \draw[latex-latex,very thin] (node5)--(node6);
  \draw[latex-latex,very thin] (node3)--(node7);
  \draw[latex-latex,very thin] (node4)--(node9);
  \draw[latex-latex,very thin] (node5)--(node8);
\end{tikzpicture}

\end{minipage}
}
\subfigure[]{
\centering
\begin{minipage}[t]{0.35\linewidth}

\begin{tikzpicture}[scale=0.18, sibling distance=0em,
  every node/.style = {scale=0.5, shape=circle, draw, align=center},
    outline/.style={draw=#1,thick,fill=#1!100}]
  every draw/.style = {scale=2}
  
  \node (node1) at (9,10) {1};
  \node (node2) at (6.5,7) {2};
  \node (node3) at (11.5,7) {3};
  \node[red] (node4) at (6.5,3) {4};
  \node[red] (node5) at (11.5,3) {5};
  \node[red] (node6) at (9,0) {6};
  \node (node7) at (14,4) {7};
  \node (node8) at (14,0) {8};
  \node (node9) at (4,0) {9};
  \draw[latex-latex,very thin] (node1)--(node2);
  \draw[latex-latex,very thin] (node2)--(node3);
  \draw[latex-latex,very thin] (node2)--(node4);
  \draw[latex-latex,very thin] (node4)--(node5);
  \draw[latex-latex,very thin] (node4)--(node6);
  \draw[latex-latex,very thin] (node5)--(node6);
  \draw[latex-latex,very thin] (node3)--(node7);
  \draw[latex-latex,very thin] (node4)--(node9);
  \draw[latex-latex,very thin] (node5)--(node8);
  \draw[-latex,red] (node1).. controls(6.5,10)..(node2);
  \draw[-latex,red] (node2).. controls(4,5)..(node4);
  \draw[-latex,red] (node4).. controls(9,5)..(node5);
  \draw[-latex,red] (node5).. controls(11.5,0)..(node6);
  \draw[-latex,red] (node6).. controls(6.5,0)..(node4);

\end{tikzpicture}

\end{minipage}
}

\subfigure[]{
\centering
\begin{minipage}[t]{0.35\linewidth}

\begin{tikzpicture}[scale=0.18, sibling distance=0em,
  every node/.style = {scale=0.55, shape=circle, draw, align=center},
    outline/.style={draw=#1,thick,fill=#1!100}]
  every draw/.style = {scale=2}
  
  \node[red] (node1) at (9,10) {1};
  \node[red] (node2) at (6.5,7) {2};
  \node[red] (node3) at (11.5,7) {3};
  \node[dashed] (node4) at (6.5,3) {4};
  \node[dashed] (node5) at (11.5,3) {5};
  \node[dashed] (node6) at (9,0) {6};
  \node (node7) at (14,4) {7};
  \node (node8) at (14,0) {8};
  \node (node9) at (4,0) {9};
  \draw[latex-latex,very thin] (node1)--(node2);
  \draw[latex-latex,very thin] (node2)--(node3);
  \draw[latex-latex,very thin] (node2)--(node4);
  \draw[latex-latex,very thin] (node4)--(node5);
  \draw[latex-latex,very thin] (node4)--(node6);
  \draw[latex-latex,very thin] (node5)--(node6);
  \draw[latex-latex,very thin] (node3)--(node7);
  \draw[latex-latex,very thin] (node4)--(node9);
  \draw[latex-latex,very thin] (node5)--(node8);
  \draw[-latex,red] (node1).. controls(6.5,10)..(node2);
  \draw[-latex,red] (node2).. controls(9,5.5)..(node3);
  \draw[-latex,red] (node3).. controls(11.5,10)..(node1);

\end{tikzpicture}

\end{minipage}
}
\subfigure[]{
\centering
\begin{minipage}[t]{0.35\linewidth}

\begin{tikzpicture}[scale=0.18, sibling distance=0em,
  every node/.style = {scale=0.5, shape=circle, draw, align=center},
    outline/.style={draw=#1,thick,fill=#1!100}]
  every draw/.style = {scale=1,thin}
  
  \node[dashed] (node1) at (9,10) {1};
  \node[dashed] (node2) at (6.5,7) {2};
  \node[dashed] (node3) at (11.5,7) {3};
  \node[dashed] (node4) at (6.5,3) {4};
  \node[dashed] (node5) at (11.5,3) {5};
  \node[dashed] (node6) at (9,0) {6};
  \node[red] (node7) at (14,4) {7};
  \node[red] (node8) at (14,0) {8};
  \node[red] (node9) at (4,0) {9};
  \draw[latex-latex,very thin] (node1)--(node2);
  \draw[latex-latex,very thin] (node2)--(node3);
  \draw[latex-latex,very thin] (node2)--(node4);
  \draw[latex-latex,very thin] (node4)--(node5);
  \draw[latex-latex,very thin] (node4)--(node6);
  \draw[latex-latex,very thin] (node5)--(node6);
  \draw[latex-latex,very thin] (node3)--(node7);
  \draw[latex-latex,very thin] (node4)--(node9);
  \draw[latex-latex,very thin] (node5)--(node8);
  \draw[-latex,red] (node7)--(node8);
  \draw[-latex,red] (node8).. controls(9,-2)..(node9);
  \draw[-latex,red] (node9).. controls(5,5) ..(node7);

\end{tikzpicture}

\end{minipage}
}

\caption{A running example of LS. Red nodes form the trading cycles, and red arrows point to their favorite neighbors. Dashed nodes denote agents who left the game.}

\label{fig:sn}

\end{figure}
\begin{table}[ht]
    \centering
    
\scalebox{0.67}{
\begin{tabular}{|c|c|c|c|c|c||c|}
\hline
& \multicolumn{3}{|c|}{t=1} & \multicolumn{2}{|c|}{t=7} & Allocation\\ \hline
$N_{out}^t$  & $\emptyset $ & $ \lbrace4,5,6\rbrace $ & $\lbrace1,2,3,4,5,6\rbrace$  & $\emptyset$ & $\lbrace7,8,9\rbrace$ & $\pi_i(\theta)$\\ \hline
$r_1'$  & $\lbrace2\rbrace$ & $ \lbrace2\rbrace $ & $-$ &$-$ &$-$ &$\pi_1(\theta)=h_2$\\ \hline
$r_2'$  & $\lbrace1,3,4\rbrace$ & $ \lbrace1,3\rbrace $ & $-$ &$-$ &$-$ &$\pi_2(\theta)=h_3$\\ \hline
$r_3'$  & $\lbrace2,7\rbrace$ & $ \lbrace2,7\rbrace $ & $-$ &$-$ &$-$&$\pi_3(\theta)=h_1$\\ \hline
$r_4'$  & $\lbrace2,5,6,9\rbrace$ &$-$  &$-$  &$-$ &$-$ &$\pi_4(\theta)=h_5$\\ \hline
$r_5'$  & $\lbrace4,6,8\rbrace$ &$-$ &$-$ &$-$ &$-$ &$\pi_5(\theta)=h_6$\\ \hline
$r_6'$  & $\lbrace4,5\rbrace$ & $-$ & $-$ &$-$ &$-$ &$\pi_6(\theta)=h_4$\\ \hline
$r_7'$  & $\lbrace3\rbrace$ & $ \lbrace3\rbrace $  & $\emptyset$ & $\lbrace8,9\rbrace$  & $-$ &$\pi_7(\theta)=h_8$\\ \hline
$r_8'$  & $\lbrace5\rbrace$ & $ \emptyset $ &   $\emptyset$ &  $\lbrace7,9\rbrace$  & $-$ &$\pi_8(\theta)=h_9$\\ \hline
$r_9'$  & $\lbrace4\rbrace$ & $ \emptyset $ &   $\emptyset$ &  $\lbrace7,8\rbrace$  & $-$ &$\pi_9(\theta)=h_7$\\ \hline
\end{tabular}
}

\centering
\caption{Updates of neighbor set for each agent. The allocation is shown in the rightmost column.}
\label{tab:neigh}
\end{table}

\begin{enumerate}
    \item Starting at $\mathcal{P}(1)$, push $\mathcal{P}(1)=1$ into the stack. 
$S_{top}=1$, $f_1(R_1)=2\notin S$. Push agent 2 into the stack. $S_{top}=2$, $f_2(R_2)=4\notin S$. Push agent 4 into the stack. $S_{top}=4$, $f_4(R_4)=5\notin S$. Push agent 5 into the stack. $S_{top}=5$, $f_5(R_5)=6\notin S$. Push agent 6 into the stack. $S_{top}=6$, $f_6(R_6)=4$, since $4 \in S$, pop off all agents from agent 6 to agent 4 to form the trading cycle $C=\{4,5,6\}$ and $N_{out}^{1}=\{4,5,6\}$. 
Remove $C$ and update $r_2'=\{1,3\}$, $r_8'=\emptyset$, $r_9'=\emptyset$.

    \item Based on the new neighbor set, $S_{top}=2$, $f_2(R_2)=3\notin S$. Push agent 3 into the stack. $S_{top}=3$, $f_3(R_3)=1$. Since $1 \in S$, pop off all agents from agent 3 to agent 1 to get $C=\{1,2,3\}$ and $N_{out}^{1}=\{1,2,3,4,5,6\}$. 
    Remove $C$ and update $r_7'=\emptyset$.

    Now that the stack is empty, update $N_{out}=N_{out}^1$ and connect the remaining neighbors of $N_{out}^1$ with each other, i.e., update $r_7'=\{8,9\}$, $r_8'=\{7,9\}$, $r_9'=\{7,8\}$.
    \item Find new $t=7$ and push $\mathcal{P}(7)=7$ into the stack. $S_{top}=7$, $f_7(R_7)=8\notin S$. Push agent 8 into the stack.  $S_{top}=8$, $f_8(R_8)=9\notin S$. Push agent 9 into the stack.  $S_{top}=9$, $f_9(R_9)=7$. Since $7 \in S$, pop off all agents from agent 9 to agent 7 to get $C=\{7,8,9\}$ and $N_{out}^{7}=\{7,8,9\}$.

    \item The stack is empty again, add $N_{out}^{7}=\{7,8,9\}$ to $N_{out}$, and now we have $N_{out}=N=\{1,2,3,4,5,6,7,8,9\}$. The mechanism terminates. The final allocation of the agents is shown in Table~\ref{tab:neigh}.
\end{enumerate}

\section{Properties of Leave and Share}\label{sec:pr}

In this section, we prove that LS is IR, IC and Stable-CC.

\begin{theorem}
For any ordering $\mathcal{P}$, LS is IR.
\end{theorem}

\begin{proof}
In LS, agent $i$ leaves only when she gets an item $h_j$. Agent $i$ can always choose herself as her favorite agent, then LS will allocate $h_i$ to $i$. Thus, LS is IR.
\end{proof}

\begin{theorem}
For any ordering $\mathcal{P}$, LS is IC.
\end{theorem}

\begin{proof}
Since each agent $i$'s type consists of two parts, her preference $\succ_i$ and her neighbor set $r_i$, we will prove misreporting neither $\succ_i$ nor $r_i$ can improve her allocation. 

\noindent \textbf{Misreport on $\succ$:} For agent $i$, we fix her reported neighbor set as $r_i'$. Her real preference is $\succ_i$ and reported preference is $\succ_i'$. Now we compare her allocation $\pi_i(( \succ_i, r_i' ), \theta_{-i}) = h_j$ with $\pi_i(( \succ_i', r_i' ), \theta_{-i}) = h_{j'}$. 

Since $\mathcal{P}$ is based on the minimum distance, which is irrelevant to agents' preferences, we only need to prove that $h_j \succeq_i h_{j'}$ for all agents for a given order.

Before $i$ is pushed into the stack, all trading cycles are irrelevant to $\succ_i'$( $i$ has not been preferred by the agents in the stack before, so $\succ_i'$ is not used at all). Thus we only consider the situation when agent $i$ is pushed into the stack and then $\succ_i'$ can decide which agent after $i$ is pushed into the stack. 


When $i$ is on the top of the stack, the next pushed agent $f_i(R_i)$ is determined by $\succ_i'$. Agent $i$ can be allocated with $h_{j'}$ only when there is a trading cycle with $i$. Assume that $h_{j'} \succ_i h_j$, i.e., misreporting $\succ_i$ gives $i$ a better item. We will show this leads to a contradiction. If $i$ reported $\succ_i$ truthfully, then $i$ would first choose $j'$ before $j$ ($j'$ is pushed into the stack first), since $i$ did not get $h_{j'}$, which means $j'$ formed a cycle $C_{j'}$ without $i$. If reporting $\succ_i'$, $i$ is matched with $j'$, then it must be the case that there exists another trading cycle $B$ which breaks the cycle $C_{j'}$. Otherwise, whenever $i$ points to $j'$, $j'$ will form the original cycle $C_{j'}$ as it is independent of $i$'s preference. The only possibility for $i$ to achieve this is by pointing her favorite agent under the false preference $\succ_i'$. By doing so, $i$ can force other agents to leave earlier with different cycles including $B$. Next, we will show that it is impossible for $B$ to break $C_{j'}$.

If $B$ can actually break $C_{j'}$, there must be an overlap between $B$ and $C_{j'}$. Assume that $x$ is the node where $B$ joins $C_{j'}$ and $y$ is the node where $B$ leaves $C_{j'}$ ($x$ and $y$ can be the same node). For node $y$, her match in $B$ and $C_{j'}$ cannot be the same (the model assumes strict preference), and no matter when $y$ is pushed into the stack, both items in $B$ and $C_{j'}$ are still there. Assume the matching in $C_{j'}$ is her favorite, then cycle $B$ will never be formed. This contradicts to $h_{j'} \succ_i h_j$, so reporting $\succ_i$ truthfully is a dominant strategy.

\noindent \textbf{Misreport on $r$:} As the above showed for any reported neighbor set $r_i'$, reporting $\succ_i$ truthfully is a dominant strategy. Next, we further show that under truthful preference report, reporting $r_i$ is a dominant strategy.
That is, for the allocation $\pi_i (( \succ_i, r_i ), \theta_{-i}) = h_j$ and $\pi_i(( \succ_i, r_i' ), \theta_{-i}) = h_{j'}$, we will show $h_j \succeq_i h_{j'}$.

Firstly, we show that the tradings before $i$ being pushed into the stack are irrelevant to $i$'s neighbor set report $r_i'$. 
For all the agents ranked before $i$ in $\mathcal{P}$, their shortest distance is smaller than or equal to $i$'s shortest distance to agent $1$, which means that their shortest paths do not contain $i$ and therefore $r_i'$ cannot change them. Thus, $r_i'$ cannot change the order of all agents ordered before $i$ in $\mathcal{P}$. In addition, agent $i$ could be a cut point to disconnect certain agents $D_i$ from agent $1$, so $r_i'$ can impact $D_i$'s distances and qualification. However, $D_i$ can only be involved in the matching after $i$ is in the stack, as others cannot reach $D_i$ without $i$. 
Hence, before $i$ is pushed into the stack, the tradings only depend on those ordered before $i$ and the agents excluding $D_i$, which are independent of $i$. In fact, the order of the agents pushed into the stack before $i$ is the same no matter what $r_i'$ is. That is, when $i$ is pushed into the stack, the agents, except for $D_i$, remaining in the game is independent of $i$.

Then when $i$ misreports $r_i$, she will only reduce her own options in the favorite agent selection. Whether $r_i'$ disconnects $D_i$ or not, reporting $r_i'$ here is equivalent to modifying $\succ_i$ by disliking neighbors in $r_i\setminus r_i'$. As we have showed, this is not beneficial for the agent. Therefore, reporting $r_i$ truthfully is a dominant strategy, i.e., $h_j \succeq_i h_{j'}$.

Put the above two steps together, we have proved that 
LS is incentive compatible. 
\end{proof}

\begin{theorem}
For any ordering $\mathcal{P}$, LS is Stable-CC.
\end{theorem}

\begin{proof}
For every $S\subseteq N$ and their item set $H_S$. Let the allocation given by LS be $\pi(\theta)$. If there exists a blocking coalition $S$, where $S \subseteq N$ is the node set of a complete component in $G(\theta)$, we have $\forall i \in S, S \subseteq r_i$. A blocking coalition $S$ suggests there exists a $z(\theta)$ such that for all  $i\in S$, $z_i(\theta) \in H_S$,  $z_i(\theta) \succeq_i \pi_i(\theta)$ with at least one $j\in S$ we have $z_j(\theta) \succ_j \pi_j(\theta)$. Therefore, for all $ j \in S$, the blocking coalition guarantees the owner of $z_j(\theta)$ and $j$ are in one trading cycle. This indicates if a trading cycle contains any agent in the coalition, all the agents in the trading cycle are in the coalition. Based on LS, $z_j(\theta) \succeq \pi_j(\theta)$ means the owner of $z_j(\theta)$ will be pushed into the stack before the owner of $\pi_j(\theta)$. Thus, the trading cycle which contains the owner of $z_j(\theta)$ and $j$ can trade by the cycle (i.e.,$\forall i \in S, z_i(\theta)$ = $\pi_i(\theta)$). This contradicts the assumption of existing $j \in S, z_j(\theta) \succ_j \pi_j(\theta)$. Hence, LS is Stable-CC.
\end{proof}

\section{Optimality Analysis}

In this section, we compare our mechanism with TTC, SWN and SWC. Since PO fails to be compatible with IC, IR in the network setting, we define a cardinal index D to measure the performance and run experiments in various random graphs to show the eminence of our mechanism. Although TTC cannot be directly applied in the social network setting, it provides an upper bound for the comparison. The lower bound is given by SWN since agents should be able to swap with their neighbors. We define $\succ_i(j)$ as the $j^{th}$ favorite item of $i$. Assuming that $h_i$ is $\succ_i(j)$ and $\pi_i(\theta)$ is $\succ_i(k)$, where $j\geq k$ for IR property, we define the ascension of $i$ as $d_i=j-k$. The average ascension of agents is defined as $D=\frac{\sum_{i\in N}d_i}{n}$. We use $D$ to measure the average improvement of agents' satisfaction in a one-sided matching mechanism.


\begin{figure}[t]
    \centering
    \includegraphics[scale=0.16]{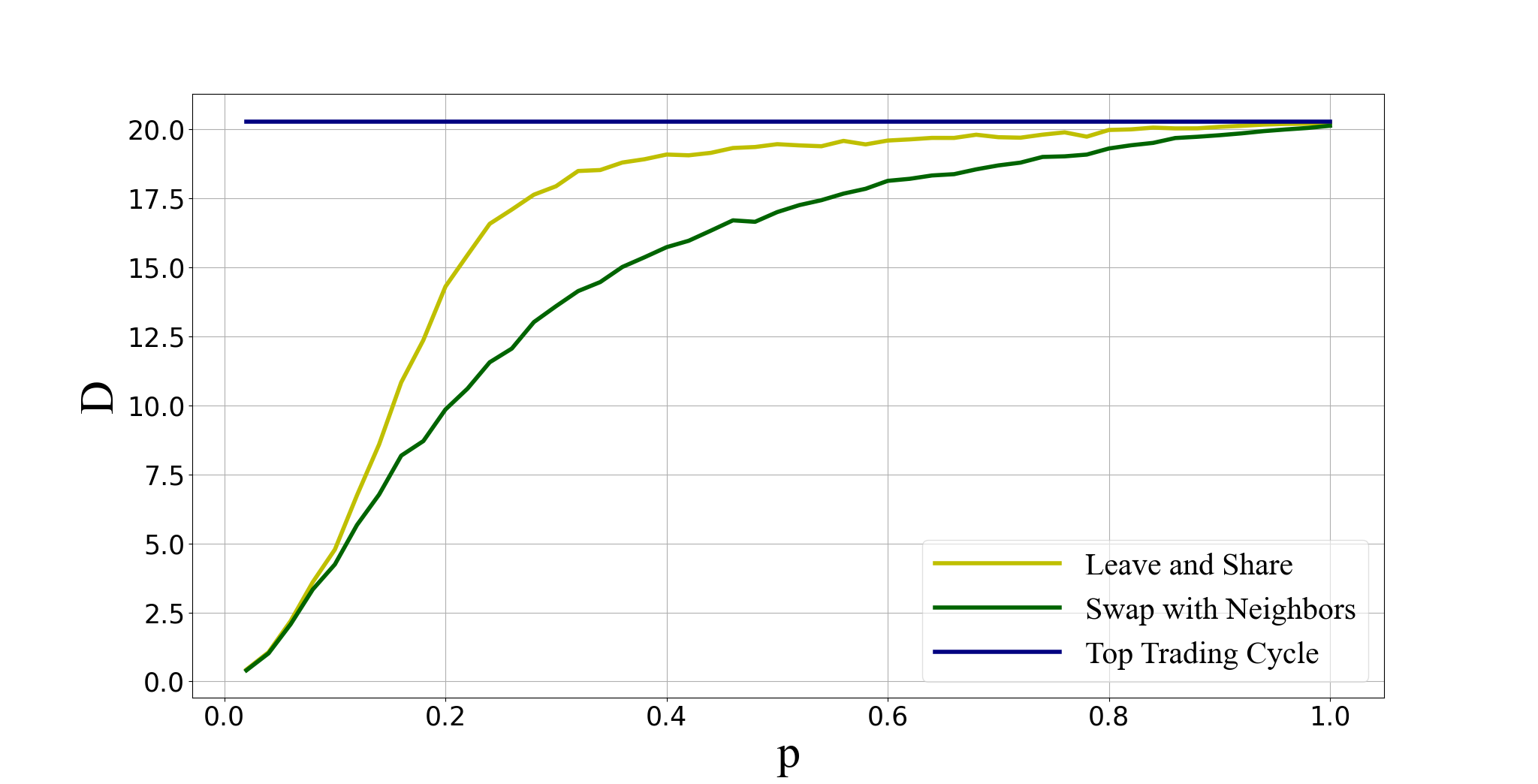}
    \caption{100 graphs generated for each p to see how D changes accordingly. The minimum scale for p is 0.02.}
    \label{fig:fn}
\end{figure}

\begin{figure}[t]
    \begin{minipage}{1\linewidth}
    \centering
            \includegraphics[scale=0.16]{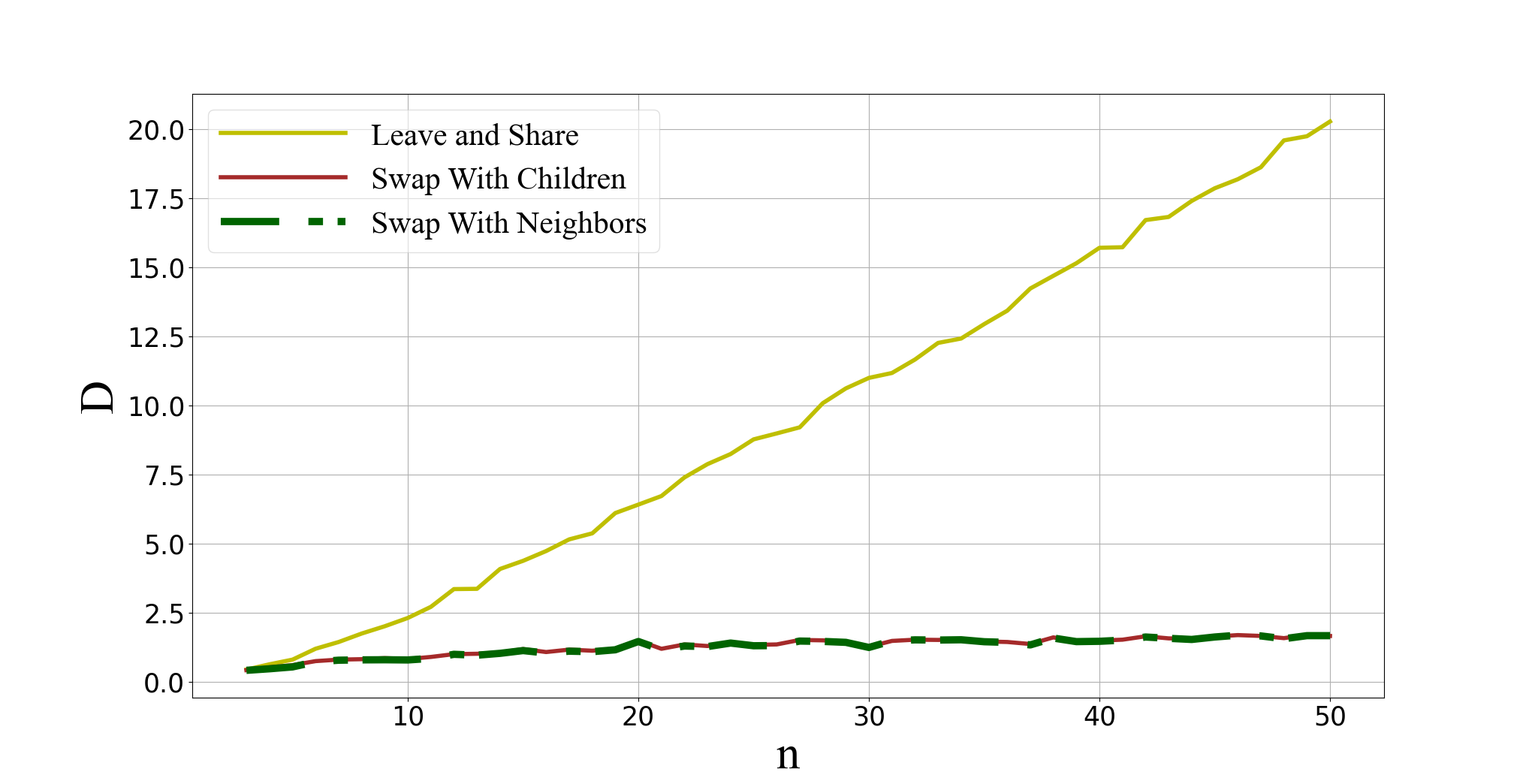}
        \caption{The tree is generated by a uniform distribution $[1,n-1]$ for each node's child node number, $n$ is the size of the tree. For LS, $n$ and $D$ are positively correlated.}
        \label{fig:fd}
    \end{minipage}
    \begin{minipage}{1\linewidth}
        \centering
            \includegraphics[scale=0.16]{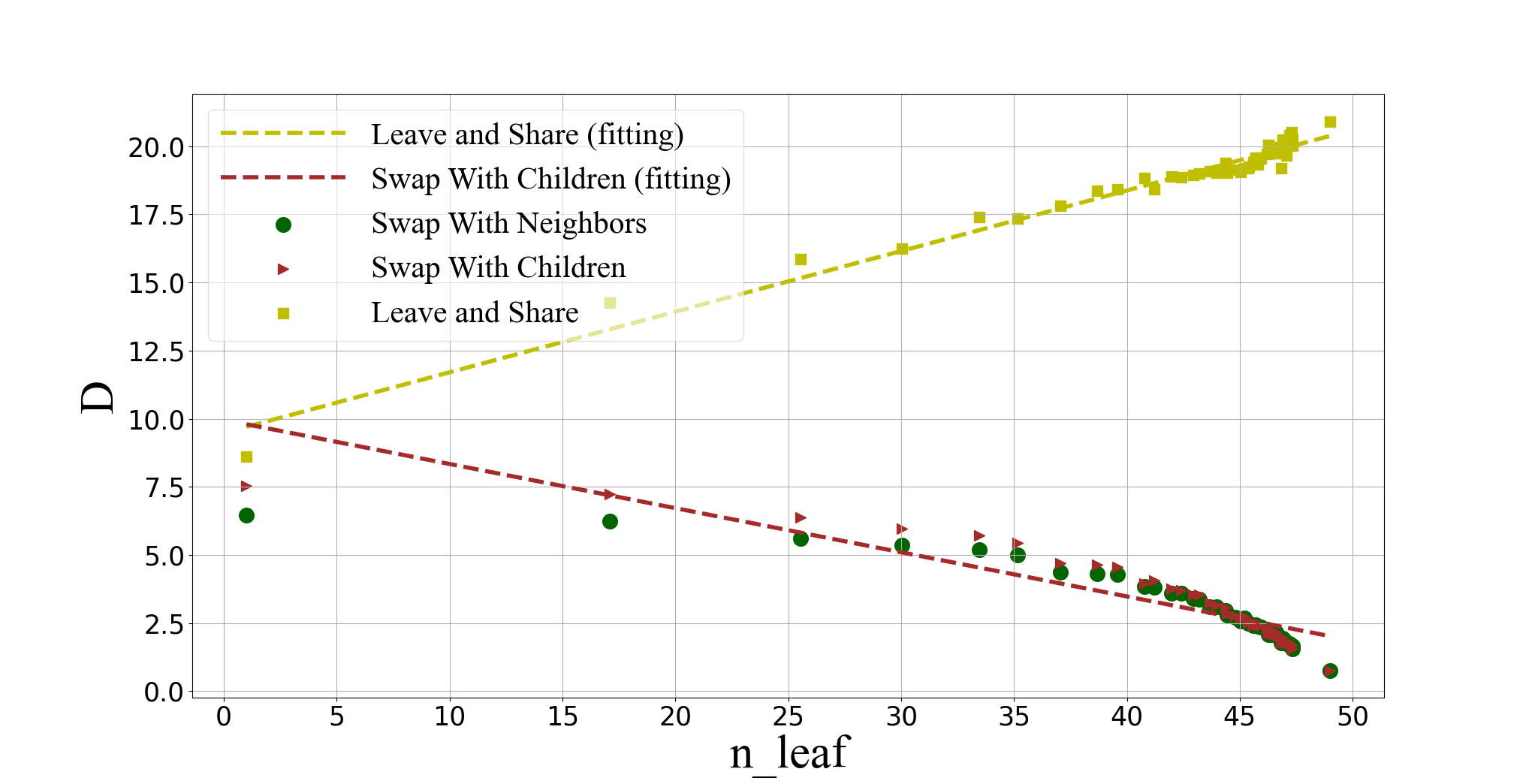}
        \caption{For LS, $n_{leaf}$ and $D$ are positive correlated. For SWN and SWC, $n_{leaf}$ and $D$ are negatively correlated.}

        \label{fig:fe}
    
    \end{minipage}
\end{figure}


Considering the lack of social network patterns in the random graphs, we adopt two special graphs, GIRG~\cite{bringmann2019geometric} and small world~\cite{watts1998collective}. Those graphs depict the cluster phenomenon in reality, and LS performs even better compared to that in the random graphs. Detailed results can be found in appendix B.


To generate random networks, we define the probability of an edge between any two nodes as $p$. A higher $p$ leads to a denser connected graph. Especially, when $p=1$, the graph is complete.
To generate a tree, we use $ub$ to represent the maximum number of child nodes for a tree, and for each node, we uniformly select an integer $i$ from $[ 1, ub ]$ as the number of the node's child nodes.  Beginning at the root node, we create $i$ child nodes for each node by a breadth-first order until the tree reaches size $n$. Agents' preferences are generated randomly from all permutations of the items.

Figure~\ref{fig:fn} shows the performances of three mechanisms.
In this figure, we generate 100 graphs of 50 nodes with fixed but randomly generated preferences and adjust $p$ to see how $D$ changes. When $p$ is close to 1, the performances of LS and SWN are close, and they are the same as TTC when $p=1$. Due to the sharing process, LS converges to TTC faster than SWN. In the other extreme case, when $p$ goes to 0, both LS and SWN have a poorer performance, because there are fewer neighbors to swap or share.

Next, we compare LS with SWC, the other extension of TTC only on trees, in two dimensions including (a) different tree sizes; (b) same tree size, but different tree structures.

In terms of tree size, we generate 100 different trees for each tree size. In Figure~\ref{fig:fd}, LS outperforms SWC and SWN significantly as the incremental tree size. Also, the results showed that SWC and SWN are quite close, because the probability of forming a cycle with more than two agents is small when applying SWC. It requires all agents except for one in a cycle to prefer their parents' item rather than their whole subtrees'. The larger the cycle is, the smaller the probability is. Thus, by comparison to SWN, the improvement of SWC is very limited, while LS takes a big advantage from sharing.

As for tree structures, we fix the tree size and use the number of leaf nodes to indicate the difference between trees.
For a certain tree with $n$ nodes, the lower bound and upper bound of $n_{leaf}$ are $1$ and $n-1$ respectively. We do the same simulation to show the relation between $D$ and $n_{leaf}$. 


In Figure~\ref{fig:fe}, we generate 100 different trees for each $ub$ from 1 to 49 and count the number of each tree's leaf nodes. Then we simulate three mechanisms on each tree. When $n_{leaf}$ is small, SWC and LS are close, because there are few neighbors to share and form a big cycle. With the increase of $n_{leaf}$, LS performs better, because sharing can match nodes in different branches, even for leaf nodes. If allocated by SWC, those leaf nodes can only get their own items.

\section{Conclusion}\label{sec:con}

In this paper, we redefined stability in social networks and showed its tightness by proving impossibilities. We proposed a novel one-sided matching protocol called Leave and Share that satisfies IC, IR and Stable-CC. Our mechanism works in all networks and significantly outperforms other mechanisms. One possible future work is to find an attainable optimality in this new setting and design mechanisms to reach it.

\appendix
\section{Properties for SWN}

\begin{theorem}
SWN is IR.
\end{theorem}

\begin{proof}
 In SWN, agent $i$ can always choose herself as her favorite agent, then SWN will allocate $h_i$ to $i$. Thus, SWN is IR.   
\end{proof}

\begin{theorem}
SWN is IC.
\end{theorem}

\begin{proof}
Since each agent $i$'s type consists of two parts, her preference $\succ_i$ and her neighbor set $r_i$, we will prove misreporting neither $\succ_i$ nor $r_i$ can improve her allocation. 

\noindent \textbf{Misreport on $\succ$:} For agent $i$, we fix her reported neighbor set as $r_i'$. Her real preference is $\succ_i$ and reported preference is $\succ_i'$. Now we suppose her allocation $\pi_i(( \succ_i', r_i' ), \theta'_{-i}) = h_{j'} \succ_i \pi_i(( \succ_i, r_i' ), \theta'_{-i}) = h_j$. 

In SWN, when $i$ truthfully reports her preference, agent $i$ is allocated $h_j$ instead of $h_{j'}$, we know that $h_{j'}$ is in a trading cycle without $i$. Let the cycle containing $j'$ be $C_{j'}$. When we fix the others preference and $i$ misreports as $(\succ'_i,r_i')$, since the trading cycle $C_{j'}$ is only determined by agents' preference in $C_{j'}$, which excludes agent $i$, $C_{j'}$ still forms. By SWN, agent $i$ cannot be allocated $h_{j'}$, which contradicts our assumption.

\noindent \textbf{Misreport on $r$:} For agent $i$, we fix her reported preference as $\succ_i'$. Her real neighbor set is $r_i$ and reported neighbor set is $r_i'$. Now we suppose her allocation $\pi_i(( \succ_i', r_i' ), \theta'_{-i}) = h_{j'} \succ_i' \pi_i(( \succ_i', r_i ), \theta'_{-i}) = h_j$. 

In SWN, each agent can only be allocated the item from her reported neighbor set $r_i'$. Therefore, we have $j \in r_i$, $j' \in r_i$ and $j' \in r_i'$. Since with true neighbor set $r_i$, agent $i$ is allocated $h_j$ instead of $h_{j'}$, we know that $h_{j'}$ is in a trading cycle without $i$. Let the cycle containing $j'$ be $C_{j'}$. According to SWN, each agent in $C_{j'}$ is pointing to her neighbor, which means their neighbor set is irrelevant to $r_i$ as well as any $r_i'$. Therefore, the trading cycle $C_{j'}$ still remains and excludes agent $i$ when agent $i$ misreports $r_i'$. By SWN, agent $i$ cannot be allocated $h_{j'}$, which contradicts our assumption.

Put the above two steps together, we have proved that SWN is incentive compatible. 
\end{proof}

\begin{theorem}
SWN is Stable-CC.
\end{theorem}

\begin{proof}
For every $S\subseteq N$ and their item set $H_S$. Let the allocation given by SWN be $\pi(\theta)$. If there exists a blocking coalition $S$, where $S \subseteq N$ is the node set of a complete component in $G(\theta)$. 

Since $S$ is the node set of a complete component, we have $\forall i \in S, S \subseteq r_i$. A blocking coalition $S$ suggests there exists a $z(\theta)$ such that for all  $i\in S$, $z_i(\theta) \in H_S$,  $z_i(\theta) \succeq_i \pi_i(\theta)$ with at least one $j\in S$ we have $z_j(\theta) \succ_j \pi_j(\theta)$. Therefore, for all $ j \in S$, the blocking coalition guarantees the owner of $z_j(\theta)$ and $j$ are neighbors. Based on SWN, $z_j(\theta) \succ_j \pi_j(\theta)$ means $j$ can always point to and get allocated $z_j(\theta)$ instead of $\pi_j(\theta)$. Thus, the trading cycle which contains the owner of $z_j(\theta)$ and $j$ can trade by the cycle (i.e.,$\forall i \in S, z_i(\theta)$ = $\pi_i(\theta)$). This contradicts the assumption of existing at least one $j \in S \subseteq r_j, z_j(\theta) \succ_j \pi_j(\theta)$. Hence, SWN is Stable-CC.
\end{proof}

\section{Additional Simulations}
Due to the lack of real-world data sets, we apply two types of networks that capture the feature of social networks.
\subsection{GIRG}
GIRG first introduced by \citeauthor{bringmann2019geometric}~\shortcite{bringmann2019geometric} identifies the structural properties of social networks. It is a scale-free random network with large clustering coefficients. We carry out the same experiments regarding the average ascension in rank $D$ as the increase of agent number $n$. Here we adopt the default parameter $\tau=2.9,\alpha=6$ in the GIRG model. We generate 100 random GIRG networks for 10 different $n$ from 5 to 50.

\begin{figure}[ht]
    \centering
    \includegraphics[scale=0.16]{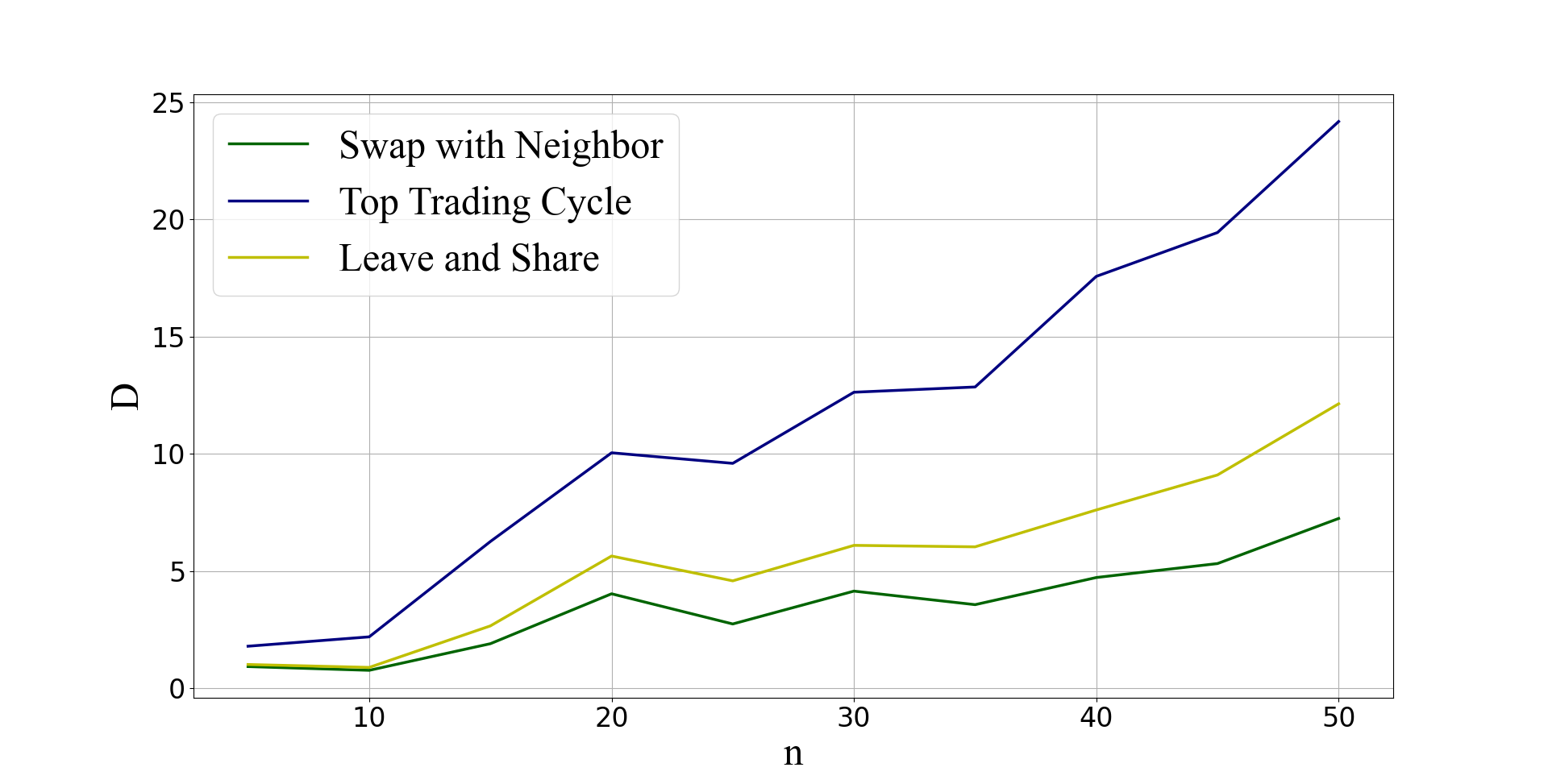}
    \caption{100 GIRG graphs generated for each $n$ to see how $D$ changes accordingly.}
    \label{fig:girg}
\end{figure}

\begin{figure}[ht]
    \centering
    \includegraphics[scale=0.16]{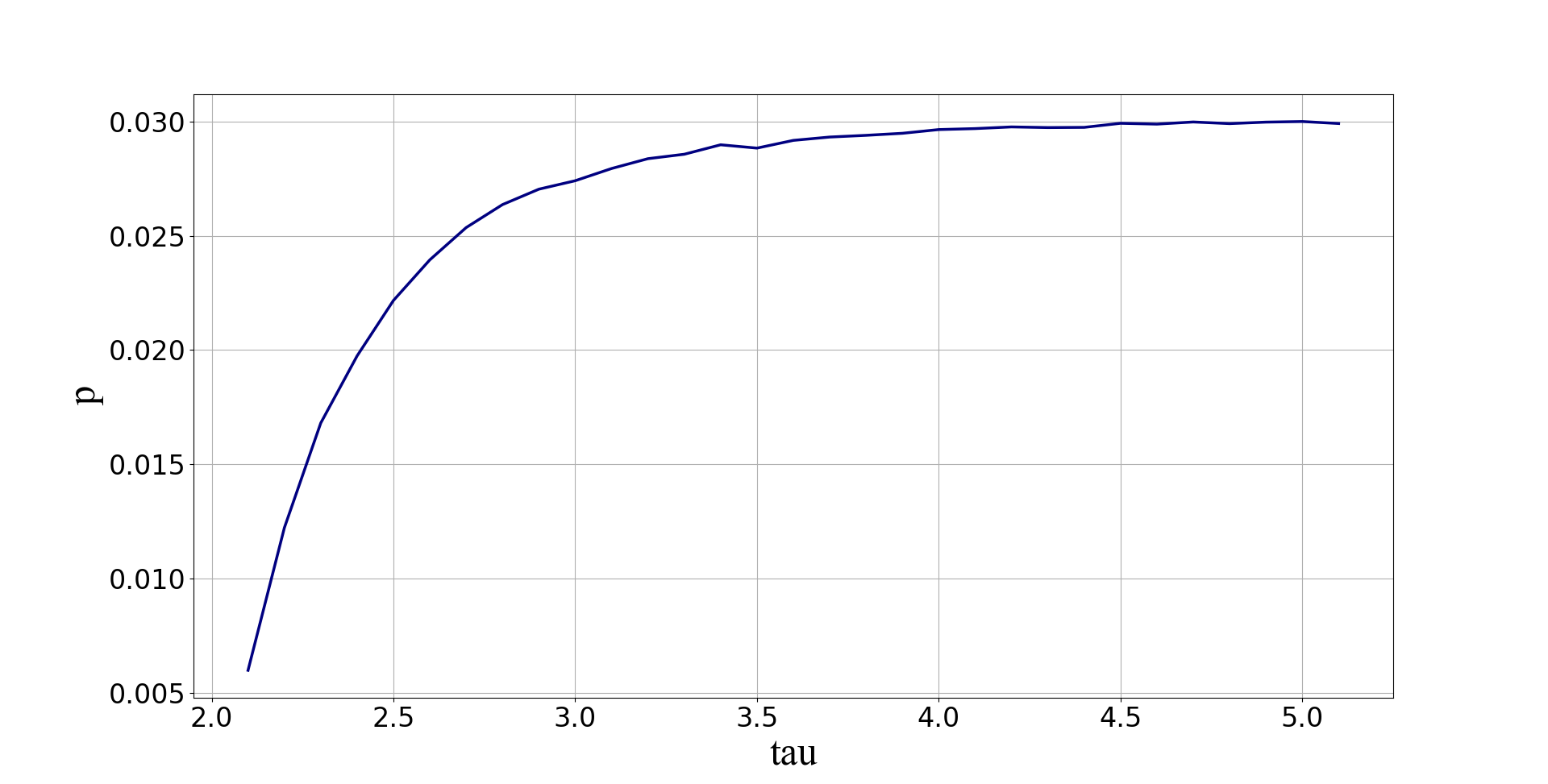}
    \caption{GIRG graph with $\alpha=6$, see the change in expected connectivity with $\tau$ for 5000 times}
    \label{fig:tau}
\end{figure}

\begin{figure}[ht]
    \centering
    \includegraphics[scale=0.16]{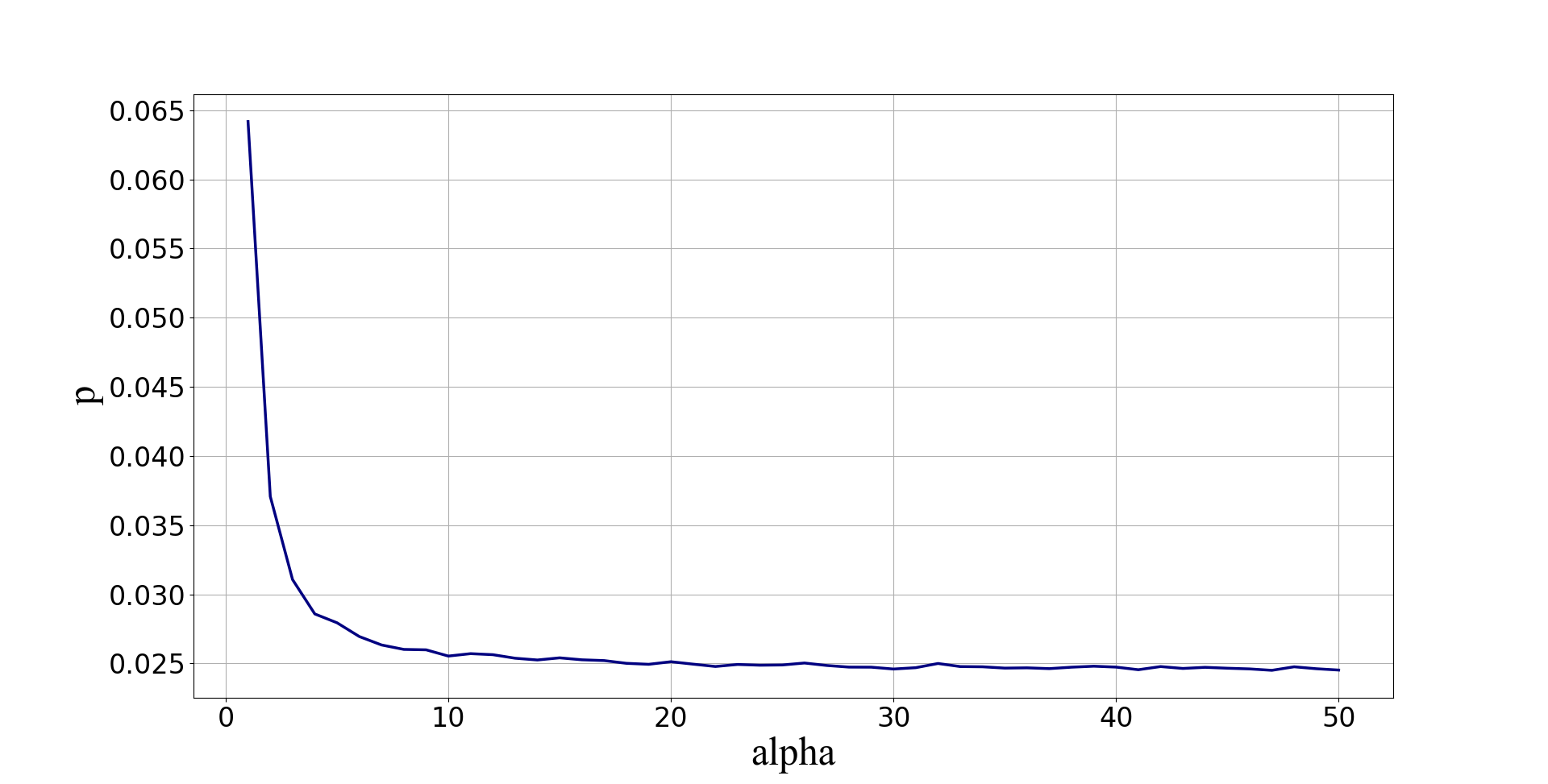}
    \caption{GIRG graph with $\tau=2.9$, see the change in expected connectivity with $\alpha$ for 5000 times}
    \label{fig:alpha}
\end{figure}

As shown in Figure~\ref{fig:girg}, LS outperforms SWN when the social network is a GIRG. Combining Figure~\ref{fig:tau} and Figure~\ref{fig:alpha}, we can see that GIRGS have a very narrow range of connectivity $p$. Also, comparing GIRGS to the randomly generated graph with same connectivity, LS's performance in GIRGS is a lot better. One explanation is that GIRGs have many clusters, which makes the sharing process much more efficient.

\subsection{Small World}
\citeauthor{watts1998collective}~\shortcite{watts1998collective} started the research regarding the small world phenomenon, and this particular network is widely used in simulations~\cite{bakshy2012social,phan2015natural,goel2009social}. In their model, the parameter $k$ indicates the expected neighbors number for each agent. We carry out the same experiments to observe the change in average ascension $D$ and $k$. Note that, we can use $p = \frac{k}{n}$ to compute the expected connectivity. We generate 100 small world networks with 50 agents and a fixed preference for 10 different $k$ from 5 to 50.

\begin{figure}[ht]
    \centering
    \includegraphics[scale=0.16]{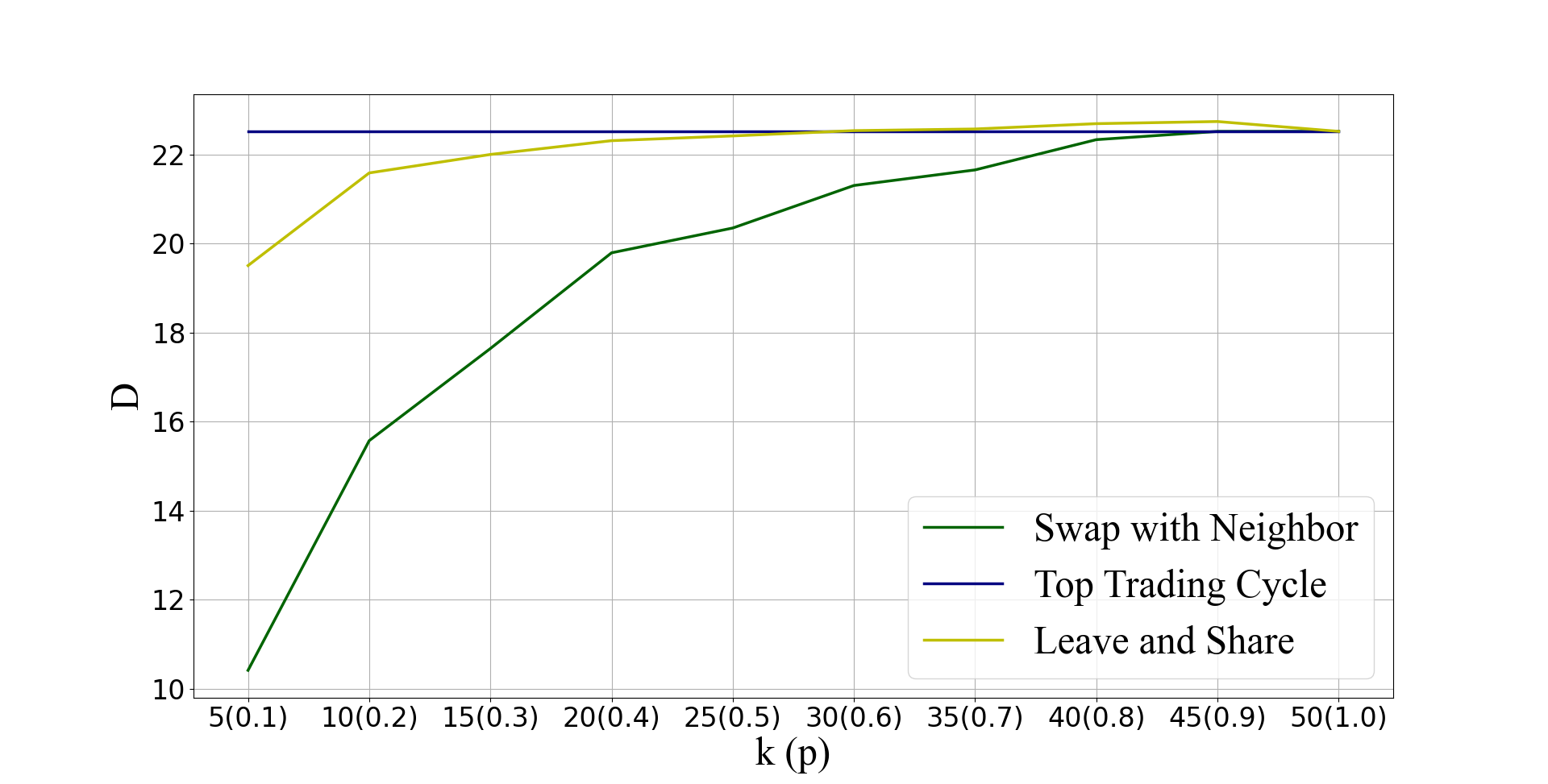}
    \caption{100 small-world graphs generated for each $k$ to see how D changes accordingly.}
    \label{fig:sm}
\end{figure}

From Figure~\ref{fig:sm}, we can see that LS performs much better in low connectivity compared with randomly generated graphs. Since the small world network model guarantees the graph to be a connected graph in extremely low connectivity, the sharing process of LS can always transform the network to be a complete graph. Therefore, the allocation given by LS is much closer to TTC.

\bibliographystyle{named}
\bibliography{ijcai23.bib}

\end{document}